\documentclass[a4paper]{article}

\usepackage[french,english]{babel}
\usepackage[T1]{fontenc}
\usepackage[utf8]{inputenc}
\usepackage{amsmath,amssymb,amsfonts,amsthm}
\usepackage{makeidx}
\usepackage{color}
\usepackage{empheq}
\usepackage{cancel}
\usepackage[top=20mm, bottom=20mm, left=20mm, right=20mm]{geometry}
\usepackage{stmaryrd} % pour \llbracket et \rrbracket.
\usepackage{enumerate}
\usepackage{hyperref} % à utiliser avec \texorpdfstring{$\mathcal{L}(\gamma)$}{Lg} dans les titres.
\usepackage{booktabs}
\usepackage[all]{xy} % pour les diagrammes commutatifs.

\usepackage[vcentermath,noautoscale]{youngtab}
\everymath{\Yboxdim10pt}
\everydisplay{\Yboxdim10pt}

\SetSymbolFont{stmry}{bold}{U}{stmry}{m}{n} % formule magique pour qu'il ne râle pas avec \boldmath.

\newtheorem{thm}{Theorem}[section]%%%%%%%%%%% Titre en gras droit, corps de texte italique %%%%%%%%%%%%
\newtheorem{conj}{Conjecture}
\newtheorem{prop}[thm]{Proposition}
\newtheorem{cor}[thm]{Corollary}
\newtheorem{lem}[thm]{Lemma}

\theoremstyle{remark}%%%%%%%%%%%%%%%%%%%%%%%%% Titre en italique, corps de texte droit %%%%%%%%%%%%%%%%%
\newtheorem{rem}{\textsc{Remark}}

\theoremstyle{definition}%%%%%%%%%%%%%%%%%%%%% Titre en gras droit, corps de texte droit %%%%%%%%%%%%%%%
\newtheorem{defi}[thm]{Definition}

\newcommand{\Cb}{\mathbb{C}}
\newcommand{\Cc}{\mathcal{C}}

\newcommand{\Ec}{\mathcal{E}}

\newcommand{\Gc}{\mathcal{G}}

\newcommand{\Ib}{\mathbb{I}}
\newcommand{\Id}{\mathrm{Id}}
\newcommand{\Jc}{\mathcal{J}}
\newcommand{\Kc}{\mathcal{K}}

\newcommand{\Mc}{\mathcal{M}}
\newcommand{\Mr}{\mathrm{M}}
\newcommand{\Mb}{\mathbb{M}}

\newcommand{\Nb}{\mathbb{N}}

\newcommand{\Nr}{\mathrm{N}}

\newcommand{\Rr}{\mathrm{R}}
\newcommand{\Rb}{\mathbb{R}}

\newcommand{\Sc}{\mathcal{S}}

\DeclareMathOperator{\Pf}{Pf}
\DeclareMathOperator{\vol}{vol}
\DeclareMathOperator{\com}{com}

\renewcommand{\d}{\ensuremath{\operatorname{d}\!}}

\newcommand{\derp}[2]{\frac{\partial#1}{\partial#2}}

\newcommand{\drond}{\partial}
\renewcommand{\vec}[1]{\overrightarrow{#1}}

\title{\textsf{\textbf{The constraint equations of Lovelock gravity theories: \\ a new \texorpdfstring{$\sigma_k$}{sigma k}-Yamabe problem}}}
\author{\textsf{Xavier Lachaume}}
\date{\textsf{\normalsize{Summer 2017}}}

\makeindex
\begin{document}
\maketitle

\begin{center}
\textsf{
Institut Denis Poisson \\
Université de Tours - UMR 7013 du CNRS \\
Parc de Grandmont - 37200 Tours - France
}

\texttt{xavier.lachaume@lmpt.univ-tours.fr}
\end{center}

\vfill

\textsc{Abstract:}
This paper is devoted to the study of the constraint equations of the Lovelock gravity theories. In the case of a conformally flat, time-symmetric, space-like manifold, we show that the Hamiltonian constraint equation becomes a generalisation of the $\sigma_k$-Yamabe problem. That is to say, the prescription of a linear combination of the $\sigma_k$-curvatures of the manifold. We search solutions in a conformal class for a compact manifold.

Using the existing results on the $\sigma_k$-Yamabe problem, we describe some cases in which they can be extended to this new problem. This requires to study the concavity of some polynomial. We do it in two ways: regarding the concavity of an entire root of this polynomial, which is connected to algebraic properties of the polynomial; and seeking analytically a concavifying function. This gives several cases in which a conformal solution exists.

At last we show an implicit function theorem in the case of a manifold with negative scalar curvature, and find a conformal solution when the Lovelock theories are close to General Relativity.

\vfill

%\tableofcontents

\newpage

\section{Introduction}

\subsection{Physics issues}

Lovelock gravity theories are one of the many attempts to generalise General Relativity (GR), in order to explain puzzling observations such as dark matter or dark energy, or to explore the consequences of string theories and AdS/CFT correspondence and the high-dimensional space-time that they suppose. They were introduced in the 1970's by D. Lovelock: the foundation papers are \cite{Love70} and \cite{Love71}.

Although these theories are well explored by the physicists, there is not yet much mathematical work on them (see eg. \cite{Patterson81} and \cite{Ge14}). One of the natural questions that arise from them is the well-posedness of the Cauchy problem, which has to be separated in two steps: the construction of suitable initial data, and the propagation of this data along some time-like vector field. The propagation problem consists in studying the hyperbolicity of the evolution equations. The seminal studies are in \cite{Tei87}, \cite{Cho88} and \cite{Cho09-692}, while recent results can be found in \cite{Reall14}, \cite{Reall15}, \cite{Will14} and \cite{Will15}. The construction of initial data consists in solving the constraint equations of the theory. This is the aim of the present paper.

The constraint equations of Lovelock theories were known since the end of the 1980's (see \cite{Tei87}, \cite{Cho88}, \cite{Cho09-692}), but they were not studied at all up to now, unlike their GR counterparts. One can find their explicit form in the particular case of Gauss-Bonnet gravity in \cite{Tor08}. More recently, the ADM decomposition of $f(\text{Lovelock})$ theories was done in \cite{Lac17-1}, which are a generalisation of Lovelock theories. In \cite{Lac17-1}, it is showed as well that in the time-symmetric case, ie. when the second fundamental form is set to $0$, the Lovelock constraint equations reduce to
\begin{equation}
\sum_{k=0}^{p_n} \alpha_k \Rb_{k} = \Ec, \label{eqcontrainte}
\end{equation}
where $n$ is the dimension of space, $p_n = \left\lfloor \frac{n+1}{2} \right\rfloor$, $\alpha_k$ are arbitrary real coefficients, $\Rb_k$ is the $k$-th Lovelock product which will further be defined and $\Ec$ is a prescribed energy. In this equation the physical dimensionality is carried by the $\alpha_k$: each $\Rb_k$ has the dimension of a $\text{length}^{-2k}$.

The GR case corresponds to
\[\begin{array}{rlrl}
\toprule
\alpha_0		& = \dfrac{-\Lambda}{\kappa},
& \alpha_1		& = \dfrac{1}{2\kappa}, \\
n		& = 3,
& \alpha_k		& = 0 \text{\quad for $k \geq 2$}, \\
\midrule
\Gc		& = \text{ the universal gravitational constant},
& \kappa	& = \dfrac{8 \pi \Gc}{c^4}, \\
\Lambda		& = \text{ the cosmological constant}. \\
\bottomrule
\end{array}\]
In the present paper we skip the physical dimensionality of the problem by choosing the units $c = 1$, $\kappa = 1/2$, so that the GR case becomes
\begin{align}
\alpha_0	&= -2\Lambda,		&\alpha_1 &= 1, \\
0		&< \Lambda \ll 1	&\alpha_k &= 0 \text{\quad for $k \geq 2$}. \nonumber
\end{align}
Although there are some investigations about a negative cosmological constant $\Lambda$, most of the cosmological models consider a small but positive one. The general Lovelock case, as for it, releases the constraints on the $\alpha_k$'s. However, in this paper we shall alternatively make two additional assumptions.

In section \ref{sectionpos}, we shall suppose that $\alpha_k \geq 0$ for $k \geq 1$. That enables us in some cases to obtain an elliptic system and that is a common physical hypothesis, although not the only one: in \cite{Zou11}, the thermodynamics of black holes is studied with $\alpha_2 < 0$; and in \cite{Kord15} or \cite{Mehdi15}, the existence of a traversable wormhole implies the same $\alpha_2 < 0$. Furthermore, the ellipticity of the problem requires to restrict the eigenvalues of the Ricci tensor to the positive cone, which implies that the scalar curvature of the space manifold is positive.

In section \ref{sectiontfi}, we shall assume the opposite: that the scalar curvature of the space manifold is negative. We shall suppose that $ \left|\alpha_k\right| \ll 1$ for $k \geq 2$. This corresponds to the GR limit of Lovelock gravity and makes possible the use of an implicit function theorem.

\subsection{Notations}

We are now going to report the notations needed to present our problem.

\subsubsection{Geometry and Lovelock products}
We write
\[\begin{array}{rll}
\toprule
(\Mc,\gamma)	& &\text{a Riemannian manifold with its metric $\gamma$}, \\
\dim \Mc	&= n, \\
2^*		&= \dfrac{2n}{n-2} &\text{the critical number in the Yamabe problem}, \\
\xi^i		& &\text{the local coordinates on $\Mc$}, \\
\gamma_{ij}	& &\text{the coefficients of $\gamma$ in $(\xi^i)$}, \\
|\gamma|	&= \det (\gamma_{ij}) &\text{the determinant of $(\gamma_{ij})$}, \\
\nabla[\gamma]		& &\text{the covariant derivative of $(\Mc,\gamma)$}, \\
\Rr[\gamma]_{ijkl}	& &\text{the Riemann tensor of $\gamma$}, \\
\Rr[\gamma]_{ij}	& &\text{the Ricci tensor of $\gamma$}, \\
\Rr[\gamma]		& &\text{the scalar curvature of $\gamma$}, \\
\Delta[\gamma]		&= \nabla[\gamma]_i \nabla[\gamma]^i &\text{the Laplacian on $\Mc$}, \\
\d v_\gamma	&= \sqrt{|\gamma|}\d^{n} \xi &\text{the volume element of $(\Mc,\gamma)$}, \\
\Cc[\gamma]	&= \left\{\hat{\gamma} = u^{\frac{4}{n-2}} \gamma \ \big| \ u \in \Cc^\infty(\Mc)\right\} &\text{the conformal class of $\gamma$.} \\
\midrule
\text{For}	& 0 \leq k \leq n, \\
\sigma_k(\lambda)	&= \displaystyle{\sum_{1 \leq i_1 < \ldots < i_k \leq n} \lambda_{i_1} \ldots \lambda_{i_k}}	&\text{the $k$-th elementary symmetric polynomial of $\lambda \in \Rb^n$,} \\
\Gamma_k &=\left\{\lambda \in \Rb^n \ \big| \ \sigma_j(\lambda) > 0 \text{ for } 0 \leq j \leq k \right\}	&\text{the $k$-th G\aa rding cone,} \\
S[\gamma]_{ij}	&= \dfrac{1}{n-2}\left[\Rr[\gamma]_{ij} - \dfrac{\Rr[\gamma]}{2(n-1)}\gamma_{ij}\right]	&\text{the Schouten tensor of $\gamma$,} \\
S[\gamma]	&= \dfrac{\Rr[\gamma]}{2(n-1)}	&\text{the trace of the Schouten tensor of $\gamma$,} \\
\lambda(\gamma)	& &\text{the eigenvalues of $S[\gamma]_j^i$,} \\
\sigma_k(\lambda(\gamma))	& &\text{the $\sigma_k$-curvature of $\gamma$, ie. the $k$-th elementary} \\
			& &\text{symmetric polynomial of the eigenvalues of $S[\gamma]_j^i$,} \\
W[\gamma]_{ijkl}	& &\text{the Weyl tensor of $\gamma$,} \\
\Sc_n(\Rb) & &\text{the set of real symmetric matrices of size $n$,} \\
\midrule
p_n		&= \left\lfloor \dfrac{n+1}{2} \right\rfloor, \\
\delta_{a_1 a_2 \ldots a_k}^{b_1 b_2 \ldots b_k}
		&= \det	\begin{pmatrix}
			\delta_{a_1}^{b_1}	& \ldots	& \delta_{a_1}^{b_k} \\
			\vdots			&		&	\vdots \\
			\delta_{a_k}^{b_1}	& \ldots	& \delta_{a_k}^{b_k}
			\end{pmatrix} &\text{the generalised Kronecker symbol}, \\
\Rb_k[\gamma]		&= \dfrac{1}{2^k}\delta_{a_1 b_1 a_2 b_2 \cdots a_k b_k}^{c_1 d_1 c_2 d_2 \cdots c_k d_k}\Rr_{c_1 d_1}^{a_1 b_1}\Rr_{c_2 d_2}^{a_2 b_2} \ldots \Rr_{c_k d_k}^{a_k b_k} &\text{the $k$-th Lovelock product}, \\
\alpha_k		&\in \Rb &\text{arbitrary fixed parameters}, \\
\midrule
\Rb_0		&= 1, \\
\Rb_1		&= \Rr, \\
\Rb_2		&= \Rr^2 - 4 \Rr_a^c\Rr^a_c + \Rr_{ab}^{cd}\Rr^{ab}_{cd}, \\
\Rb_3		&= \Rr^3 + 2\Rr_{ab}^{cd}\Rr_{cd}^{ef}\Rr_{ef}^{ab} + 3\Rr\Rr_{ab}^{cd}\Rr_{cd}^{ab} \\
		& + 8\Rr_{ab}^{cf}\Rr_{cd}^{eb}\Rr_{ef}^{ad} -12\Rr\Rr_a^b\Rr_b^a +16\Rr_a^b\Rr_b^c\Rr_c^a \\
		& - 24\Rr_a^b\Rr_{cd}^{ae}\Rr_{be}^{cd} + 24\Rr_a^b\Rr_c^d\Rr_{bd}^{ac}, \\
\vdots \\
\Rb_k		&= 0 &\text{for $k > p_n$, because of the anti-symmetries of $\Rr_{ab}^{cd}$}. \\
\bottomrule
\end{array}\]
We use the Einstein summation convention. The coefficients are written in the local coordinates $(\xi^i)$. The dependence in $\gamma$ is omitted if not necessary.

\begin{rem}
The G\aa rding cones $\Gamma_k$ were introduced by L. G\aa rding for his work on the hyperbolic polynomials, see \cite{Gar59}. They satisfy the following inclusions:
\[\Gamma_1 \supset \Gamma_2 \supset \ldots \supset \Gamma_{n-1} \supset \Gamma_n = \left(\Rb_+^*\right)^n.\]
A metric $\gamma$ which verifies $\lambda(\gamma) \in \Gamma_k$ is called $k$-admissible.
\end{rem}
Let us now focus on the elementary symmetric polynomials.

\subsubsection{Differential calculus and linear algebra}
When we consider quantities living in $\Rb^n$, which is a flat space, we can write indifferently in lower or upper indices their coordinates. In this case, we use the following notations.
\begin{defi}
Let $n \geq 1$, and $\Omega$ an open set of $\Rb^n$. For $f, g \in \Cc^2(\Omega \rightarrow \Rb)$, $A, B \in \Mc_n(\Rb)$, we write
\[\begin{array}{rll}
\toprule
f_i &= \drond_{x_i} f, \\
f_{ij} &= \drond_{x_i}\drond_{x_j} f, \\
\midrule
D(f) &= (f_i)_{1 \leq i \leq n} \in \Cc^1(\Omega \rightarrow \Rb^n) &\text{ the gradient of } f, \\
H(f) &= (f_{ij})_{1 \leq i,j \leq n} \in \Cc^0(\Omega \rightarrow \Sc_n(\Rb)) &\text{ the hessian of } f, \\
\midrule
f \equiv g &\Longleftrightarrow \text{ $f$ and $g$ have the same sign on $\Omega$}, \\
A \equiv B &\Longleftrightarrow \text{ there exists some $\lambda > 0$ such that $A = \lambda B$}. \\
\bottomrule
\end{array}\]

We assimilate a real polynomial with its polynomial function on $\Rb$.

For $x, y \in \Rb^n$, we denote by
\[\begin{array}{rl}
\toprule
 ^t\!x & \text{ the transpose of } x, \\
(x|y) = ^t\!\!xy = ^t\!\!yx = \sum_i x_i y_i & \text{ the scalar product of $x$ and $y$}, \\
x \otimes y = x ^ty = (x_iy_j)_{ij} \in \Mc_n(\Rb) & \text{ the tensor product of $x$ and $y$}, \\
\midrule
\Ib &= (1, \ldots, 1) \in \Rb^n, \\
\Id &= \begin{pmatrix} 1 & & 0 \\ & \ddots & \\ 0 & & 1 \end{pmatrix} \in \Mc_n(\Rb), \\
%x'_i = (x_1, \ldots, \cancel{x_i}, \ldots, x_n) \in \Rb^{n-1} & \text{ the vector $x$ without its $i$-th coordinate}, \\
%\midrule
%(e_i)_{1 \leq i \leq n} & \text{ the canonic basis of } \Rb^n. \\
\bottomrule
\end{array}\]
%We write $\Ib = (1, \ldots, 1) \in \Rb^n$, and keep the same notation whatever the dimension is when there is no ambiguity.
We call \emph{diagonal} the open diagonal half-axis
\[\Rb_+^* \cdot \Ib = \left\{(t, \ldots, t) \in \Gamma_n \ | \ t > 0\right\}.\]

For $\vec{a} = (a_0, a_1, \ldots, a_p) \in \Rb^{p+1}$, we set
\[\begin{array}{rrll}
\toprule
&f_{\vec{a}} &= a_0 + a_1 \sigma_1 + \ldots + a_p \sigma_p &\in \Rb[X_1, \ldots, X_n], \\
f_{\vec{a}}(X, \ldots, X) &= \bar{f}_{\vec{a}} &= \displaystyle{a_0 \binom{n}{0} + a_1 \binom{n}{1} X + \ldots + a_p \binom{n}{p} X^p} &\in \Rb[X], \\
\bottomrule
\end{array}\]
\end{defi}

Finally, let us introduce a set that is studied in more details in \cite{Lac17-3}:
\[\Kc_n^p = \left\{\vec{a} \in \Rb^{p+1} \quad \big| \quad f_{\vec{a}}^{1/p} \text{ is concave on } \Gamma_n\right\}. \]

\subsubsection{\texorpdfstring{$n+1$}{n+1} decomposition of Lovelock gravity}

The projection of the field equations of some gravity theory onto a $n+1$ space-like foliation of space-time is called the $n+1$ decomposition -- or ADM decomposition -- of this theory. It splits the field equations into constraint equations and dynamical equations. Our problem finds its origin in the $n+1$ decomposition of the Lovelock field equations, and more specifically in the constraint equations. These projected equations are already in the literature (see \cite{Tei87}, \cite{Cho88}, \cite{Cho09-692}). Let us introduce a few more tensors and physical quantities so as to present them.
\[\begin{array}{rll}
\toprule
K_{ij}		& &\text{the second fundamental form of the embedding} \\
& &\quad \text{of $\Mc$ into a $n+1$-dimensional space-time}, \\
\Mr_{ab}^{cd}	&= \Rr_{ab}^{cd} + K_a^c K_b^d - K_a^d K_b^c &\text{the space-time Riemann tensor projected four} \\
& &\quad \text{times on $\Mc$ according to Gauss equation}, \\
\Nr^{ab}_c &= -\nabla^a K^b_c + \nabla^b K^a_c &\text{the space-time Riemann tensor projected three} \\
& &\quad \text{times on $\Mc$ according to Codazzi identity}, \\
\midrule
\Mb_k		&= \dfrac{1}{2^k}\delta_{a_1 b_1 \cdots a_k b_k}^{c_1 d_1 \cdots c_k d_k}\Mr_{c_1 d_1}^{a_1 b_1} \ldots \Mr_{c_{k-1} d_{k-1}}^{a_{k-1} b_{k-1}} \Mr_{c_k d_k}^{a_k b_k}, \\
\Nb_{(k)}^i	&= \dfrac{1}{2^k}\delta_{a_1 b_1 \cdots a_k b_k}^{c_1 d_1 \cdots c_k i}\Mr_{c_1 d_1}^{a_1 b_1} \ldots \Mr_{c_{k-1} d_{k-1}}^{a_{k-1} b_{k-1}} \Nr_{c_k}^{a_k b_k}, \\
\midrule
\Ec & &\text{the energy density}, \\
\Jc^i & &\text{the momentum density}. \\
\bottomrule
\end{array}\]

Then the constraint equations of Lovelock gravity are
\begin{align}
\sum_{k=0}^{p_n} \alpha_k \Mb_{k}	&= \Ec,	&\text{(Hamiltonian constraint)} \label{eqcontrainte1LL} \\
\sum_{k=0}^{p_n} 2 k \alpha_k \Nb_{(k)}^i	&= \Jc^i,	&\text{(momentum constraint)} \label{eqcontrainte2LL}
\end{align}
where $\Ec$ and $\Jc^i$ are prescribed, so the equations concern $(\gamma,K)$. An explicit example for Gauss-Bonnet gravity can be found in \cite{Tor08}.

$(\ref{eqcontrainte1LL},\ref{eqcontrainte2LL})$ is an under-determined system, and the common way to solve it is to fix a background metric and to search a solution in a given conformal class.

\subsection{Assumptions}

In this paper we shall make three strong assumptions:
\begin{enumerate}
\item We restrict ourselves to the \textbf{time-symmetric case}, ie. we assume that $K_{ij} = 0$. The constraint equations $(\ref{eqcontrainte1LL},\ref{eqcontrainte2LL})$ then become:
\begin{align}
\sum_{k=0}^{p_n} \alpha_k \Rb_{k}	&= \Ec, \label{eqcontrainte1LLtsym} \\
0					&= \Jc^i. \label{eqcontrainte2LLtsym}
\end{align}
$(\ref{eqcontrainte2LLtsym})$ is a physical necessary condition for a time-symmetric decomposition to be possible. But $(\ref{eqcontrainte1LLtsym})$ is a new geometrical equation. When all but one of the $\alpha_k$'s are taken to be $0$, $(\ref{eqcontrainte1LLtsym})$ is only a curvature prescription equation.

Let us remark that \eqref{eqcontrainte1LLtsym} is not about the Lovelock polynomial that is present in the action describing the Lovelock theories, but about the \textbf{projection} of this polynomial on a space-like hypersurface $\Mc$. The Lovelock polynomial is built on the space-time Lorentzian metric, while \eqref{eqcontrainte1LLtsym} is built on a space-like Riemannian metric. The fact that this projection has the same structure as a Lovelock polynomial is due to the form of $\Mr_{a b}^{c d}$ and the hypothesis of time-symmetry.

\item We shall assume that $\Mc$ is \textbf{locally conformally flat}. Indeed, let us recall that the Riemann tensor can be decomposed so:
\[\Rr_{ab}^{cd} = W_{ab}^{cd} + \left(S_a^c \delta_b^d + S_b^d \delta_a^c - S_a^d \delta_b^c - S_b^c \delta_a^d\right).\]
So in the locally conformally flat case, ie. if $W_{abcd} = 0$, it can be shown (see for example \cite{Ge14}) that for $0 \leq k \leq p_n$,
\[\Rb_{k} = 2^k k!\dfrac{(n-k)!}{(n-2k)!}\sigma_k(\lambda(\gamma)).\]
Hence, $(\ref{eqcontrainte1LLtsym})$ becomes
\begin{equation}
\sum_{k=0}^{p_n} \alpha_k 2^k k!\dfrac{(n-k)!}{(n-2k)!}\sigma_k(\lambda(\gamma)) = \Ec. \label{eqalpha}
\end{equation}
That is why we set
\begin{equation}
%E &:= \frac{\Ec}{2\Lambda}, \\
a_k := \alpha_k 2^k k!\dfrac{(n-k)!}{(n-2k)!} \label{defak}
\end{equation}
which has the same sign as $\alpha_k$.

\item We suppose that $\Mc$ is compact. Considering a finite volume manifold is indeed a necessary hypothesis for all the works that we shall use.
\end{enumerate}

Now, here is the main problem of this paper: we assume that $(\Mc,\gamma)$ is locally conformally flat and compact, and that its embedding in a surrounding space-time is time-symmetric; this implies that $\eqref{eqcontrainte2LLtsym}$ is true. Then the constraint equations of Lovelock theories reduce to
\begin{equation}
\sum_{k=0}^{p_n} a_k \sigma_k(\lambda(\gamma)) = \Ec. \label{eqa}
\end{equation}

If all the $a_k$ but one are set to $0$, \eqref{eqa} is a problem of $\sigma_k$-curvature prescription. As well as for the constraint equations of GR, it is under-determined, so we choose to search a solution in a given conformal class: this is the $\sigma_k$-Yamabe problem, a generalisation of the Yamabe problem.

For $k=1$, $\sigma_1(\lambda(\gamma))$ is the scalar curvature of $\gamma$, so we recover the famous Yamabe problem.

For $k=n$, although we are not concerned by this case, $\sigma_n(\lambda(\gamma))$ would be the determinant of the Schouten tensor, hence a sort of Monge-Amp\`ere equation, which gave birth to a flourishing literature.

The other cases are an interpolation between those two extremal cases. As explained in \cite{Lab08}, here are some cases for which a solution of the $\sigma_k$-Yamabe problem has recently been found:
\begin{itemize}
\item 2000, seminal work from J. A. Viaclovsky \cite{Via00}, \cite{Via01};
\item 2003, locally conformally flat manifold, P. Guan and G. Wang \cite{Guan03};
\item 2003, locally conformally flat manifold, A. Li and Y. Li \cite{Li03OnSome}, \cite{Li05}, extended in \cite{Li03AFully};
\item 2007, manifold with a nonempty boundary, B. Guan \cite{Guan07};
\item 2007, $n/2 < k$, admissible metric, M. Gursky and J. Viaclovsky \cite{Gursky07};
\item 2007, $2 \leq k \leq n/2$, admissible metric and variational problem, W. Sheng, N. Trudinger and X. Wang \cite{Sheng07}.
\end{itemize}

However our equation \eqref{eqa} is not the prescription of a single $\sigma_k$-curvature, but of an arbitrary linear combination of $\sigma_k$-curvatures. This is a new sort of $\sigma_k$-Yamabe problem, that we shall now handle.

In the following we shall release the constraint on $p_n$, and consider \eqref{eqa} as a special case of this generic curvature prescription equation:
\begin{defi}
For $\vec{a} = (a_0, a_1, \ldots, a_p) \in \Rb^{p+1}$, we call $(\ref{eqfa},\vec{a},\Ec)$ the following problem: to find $\hat{\gamma} \in \Cc[\gamma]$ a solution of
\begin{equation}
f_{\vec{a}}(\lambda(\hat{\gamma})) = \Ec. \label{eqfa}
\end{equation}
\end{defi}
% We let the new $a_0$ vary in order to present more general results; this would for instance enable to consider a possibly negative cosmological constant.

\section{Concavity} \label{sectionpos}

Most of the existing solutions of the $\sigma_k$-Yamabe problem use the fact that the locally conformally flat $\sigma_k$-Yamabe problem is variational, as well as the classical Yamabe problem. This property is defined in the following theorem shown by J. Viaclovsky.
\begin{thm}[\cite{Via00}]
Let $\hat{\gamma} \in \Cc[\gamma]$, such that $\vol(\hat{\gamma}) = 1$. Then for $k \neq n/2$, $\hat{\gamma}$ satisfies
\begin{equation}
\sigma_k(\lambda(\hat{\gamma})) = \text{constant} \label{eqkhat-cst}
\end{equation}
if and only if it is a critical point of the functional
\[g \longmapsto \int_\Mc \sigma_k(\lambda(g)) \d v_g.\]
\end{thm}

However, our new problem $(\ref{eqfa},\vec{a},\Ec)$ is not variational anymore. We have to use more general results, such as the following theorem from A. Li and Y. Y. Li.
\begin{thm}[\cite{Li03AFully}] \label{thmLi}
Let $\Gamma \subset \Rb^n$ and $f \in \Cc^\infty(\Gamma) \cap \Cc^0(\overline{\Gamma})$ such that
\begin{enumerate}
\item	\begin{enumerate}
	\item $\Gamma$ is an open convex cone with vertex at the origin,
	\item $\Gamma_n \subset \Gamma \subset \Gamma_1$, \label{1b}
	\item $\Gamma$ is symmetric in the $\lambda_i$;
	\end{enumerate}
\item	\begin{enumerate}
	\item $f$ is concave and symmetric in the $\lambda_i$, \label{2a}
	\item $f = 0$ on $\drond \Gamma$,
	\item $\forall~1 \leq i \leq n$, $\drond_i f > 0$ on $\Gamma$,
	\item $\forall~\lambda \in \Gamma$, $\underset{s \rightarrow \infty}{\lim} f(s\lambda) = \infty$.
	\end{enumerate}
\end{enumerate}
And let $(\Mc,\gamma)$ be a compact locally conformally flat manifold without boundary satisfying
\[\lambda(\gamma) \in \Gamma \qquad \text{on } \Mc.\]

Then there exists some $u \in \Cc^\infty(\Mc)$, $u>0$, such that $\hat{\gamma} := u^{\frac{4}{n-2}} \gamma \in \Cc[\gamma]$ satisfies
\begin{align}
f(\lambda(\hat{\gamma})) &= 1 \qquad \text{and} \label{fhat} \\
\lambda(\hat{\gamma}) &\in \Gamma \qquad \text{on } \Mc.
\end{align}
\end{thm}

\begin{rem}
The hypothesis on the G\aa rding cones $(\ref{1b})$ is an admissibility requirement that guarantees the ellipticity of the equation \eqref{fhat}. As was said in introduction, it implies that $\sigma_1(\lambda(\gamma)) > 0$, ie. the manifold has a positive scalar curvature.
\end{rem}

\begin{rem}
For the prescription of an \emph{arbitrary} $\sigma_k$-curvature, the current results use the variational nature of the $\sigma_k$-Yamabe problem. The more general result of A. Li and Y. Y. Li does not need a variational problem, but loses the opportunity of prescription. It only enables to fix the value of $f$ to be constant.
\end{rem}

\begin{rem}
The most important hypothesis of this theorem is $(\ref{2a})$, the concavity of $f$. This originates from the following astonishing concavity property of $\sigma_k^{1/k}$; a proof of this result can be found in \cite{Marcus57} with algebraic arguments, or in \cite{Caf85} using the theory of hyperbolic polynomials of L. G\aa rding.
\end{rem}

\begin{lem}[\cite{Marcus57}, \cite{Caf85}]
Let $\lambda, \mu \in \Gamma_k$, $1 \leq k \leq n$. Then
\[\sigma_k^{1/k}(\lambda + \mu) \geq \sigma_k^{1/k}(\lambda) + \sigma_k^{1/k}(\mu),\]
with equality if and only if $\lambda /\!\!/ \mu$ (or $k=1$).
\end{lem}
The homogeneity of $\sigma_k^{1/k}$ leads to the concavity on $\Gamma_k$, and enables to take $f = \sigma_k^{1/k}$ in Theorem \ref{thmLi}.

\subsection{Algebraic arguments of concavity}

\subsubsection{Real-rootedness of \texorpdfstring{$\bar{f}_{\protect\vec{a}}$}{f}}

%$\bar{f}_{\overset{b}{a}}$
The question is, is the same true for $f_{\vec{a}}$? Obviously not: it depends on the coefficients $a_k$. We lose the essential homogeneity property, and we have to introduce a peculiar set:
\begin{defi}
\begin{equation}
\Kc_n^p = \left\{\left(a_0, a_1, \ldots, a_p\right) \in \left(\Rb_+\right)^{p+1} \quad \Bigg| \quad x \in \Gamma_n \longmapsto \left(\sum_{k=0}^{p} a_k \sigma_k(x)\right)^{1/p} \text{ is concave}\right\}
\end{equation}
\end{defi}
The precise determination of $\Kc_n^p$ is an algebraic problem, and is handled in a specific paper \cite{Lac17-3}. It is still an open problem for $p > 2$. Using this set, we shall be able to apply the Theorem \ref{thmLi} to
\[(\Gamma, f) = \left(\Gamma_n, \left(f_{\vec{a}}\right)^{1/p}\right).\]
At this point we have to precise two things.
\begin{itemize}
\item The concavity results of \cite{Lac17-3} apply only to polynomials with non-negative coefficients. Yet the most studied physical case is
\[\alpha_0 = a_0 = -2\Lambda < 0.\]
So we shall change the variable:
\begin{equation}
a_0 \longrightarrow a_0 + 2\Lambda
\end{equation}
and study $(\ref{eqfa}, \vec{a}, \Ec + 2 \Lambda)$. Nevertheless we shall let $a_0$ vary in order to present general results, eg. possibly considering a negative cosmological constant.

\item The Theorem \ref{thmLi} only enables to prescribe a constant curvature. So we shall assume that $\Ec \geq 0$ is constant -- which is for example the case in vacuum, where $\Ec = 0$ -- and renormalise the coefficients: for all $0 \leq k \leq p$,
\begin{equation}
a_k \longrightarrow \dfrac{a_k}{\Ec + 2 \Lambda}.
\end{equation}
\end{itemize}
Hence we consider the problem $(\ref{eqfa},\vec{a},1)$.

\begin{thm}
Let us suppose that $\Ec = 0$, ie. let us consider the vacuum case. Then $(\ref{eqfa},\vec{a},1)$ is equivalent to
\begin{equation}
\left(f_{\vec{a}}(\lambda(\gamma))\right)^{1/p} = 1.
\end{equation}
Thereafter, if
\begin{itemize}
\item $\vec{a} = (a_0, a_1, \ldots, a_{p}) \in \Kc_n^{p}$ and
\item $\lambda(\gamma) \in \Gamma_n$,
\end{itemize}
then there exists some $u \in \Cc^\infty(\Mc)$, $u>0$, such that $\hat{\gamma} := u^{\frac{4}{n-2}} \gamma \in \Cc[\gamma]$ satisfies
\begin{align}
\left(\sum_{k=0}^{p}a_k \sigma_k(\lambda(\hat{\gamma}))\right)^{1/p} &= 1 \qquad \text{and} \\
\lambda(\hat{\gamma}) &\in \Gamma_n \qquad \text{on } \Mc,
\end{align}
thus $(\ref{eqfa},\vec{a},1)$ is solved.
\end{thm}

\begin{cor}
If
\begin{itemize}
\item $(0, a_1, \ldots, a_{p_n}) \in \Kc_n^{p_n}$ and
\item $\lambda(\gamma) \in \Gamma_n$,
\end{itemize}
then the vacuum constraint equation \eqref{eqa} has a solution.
\end{cor}

In \cite{Lac17-3}, we explicit some particular sets of coefficients that belong to $\Kc_n^{p}$:
\begin{thm}[\cite{Lac17-3}] Let $n \in \Nb$. \label{thmLac17-3}
\begin{enumerate}[a)]
\item For $a_0, a_1, a_2 \geq 0$,
\begin{align*}
\left(a_0, a_1, a_2\right) \in \Kc_n^2 & \quad \Longleftrightarrow \quad a_0 + a_1 \binom{n}{1} X + a_2 \binom{n}{2} X^2 \text{ is real-rooted} \\
& \quad \Longleftrightarrow \quad n a_1^2 - 2(n-1)a_0a_2 \geq 0.
\end{align*}

\item For $a_0, a_1, \ldots, a_p \geq 0$,
\[a_0, a_1, \ldots, a_p \text{ are the highest-degrees coefficients of some real-rooted polynomial} \ \Longrightarrow \ \left(a_0, a_1, \ldots, a_p\right) \in \Kc_n^p.\]

\item (Kurtz's criterion) For $a_0, a_1, \ldots, a_p \geq 0$,
\[4 a_{k-1} a_{k+1} < a_k^2 \quad \forall~1 \leq k \leq p-1 \quad \Longrightarrow \quad \left(a_0, a_1, \ldots, a_p\right) \in \Kc_n^p.\]

\item (Walsh's theorem) For $a_0, a_1, \ldots, a_p \geq 0$, \label{thmWalsh}
\[\bar{f}_{\vec{a}} = a_0 \binom{n}{0} + a_1 \binom{n}{1} X + \ldots + a_p \binom{n}{p} X^p \text{ is real-rooted} \quad \Longrightarrow \quad \left(a_0, a_1, \ldots, a_p\right) \in \Kc_n^p.\]
\end{enumerate}
\end{thm}

\begin{cor} Let $\vec{a} = (0, a_1, \ldots, a_{p_n}) \in \left(\Rb_+\right)^{p_n+1}$. \label{corsolve}
\begin{enumerate}[a)]
\item Let us assume that $a_k = 0$, for $3 \leq k \leq p_n$. This corresponds to Gauss-Bonnet gravity theories. Then
\[(0, a_1, a_2) \in \Kc_n^2,\]
so the Theorem \ref{thmLi} applies to
\[(\Gamma, f) = \left(\Gamma_n, \sqrt{a_1 \sigma_1 + a_2 \sigma_2}\right)\]
and there exists a solution to $(\ref{eqfa},\vec{a},1)$. \label{corsolvea}

\item Let us assume that there exist some $m \geq 0$ and $R \in \Rb_+[X]$, $\deg R \leq m-1$, such that $R + a_1 X^{m+1} + \ldots + a_{p_n} X^{m+p_n}$ is real-rooted. That is to say, $(a_1, \ldots, a_{p_n})$ are the highest degree coefficients of some real-rooted polynomial. Then the Theorem \ref{thmLi} applies to
\[(\Gamma, f) = \left(\Gamma_n, \left(a_1 \sigma_1 + \ldots + a_{p_n} \sigma_{p_n}\right)^{1/p_n}\right)\]
and there exists a solution to $(\ref{eqfa},\vec{a},1)$. \label{corsolveb}

\item Let us assume that for all $1 \leq k \leq p_n - 1$, $4 a_{k-1} a_{k+1} < a_k^2$. Then the Theorem \ref{thmLi} applies to
\[(\Gamma, f) = \left(\Gamma_n, \left(a_1 \sigma_1 + \ldots + a_{p_n} \sigma_{p_n}\right)^{1/p_n}\right)\]
and there exists a solution to $(\ref{eqfa},\vec{a},1)$. \label{corsolvec}

\item Let us assume that $\bar{f}_{\vec{a}} = a_1 \binom{n}{1} X + \ldots + a_{p_n} \binom{n}{p_n} X^{p_n}$ is real-rooted. Then the Theorem \ref{thmLi} applies to
\[(\Gamma, f) = \left(\Gamma_n, \left(a_1 \sigma_1 + \ldots + a_{p_n} \sigma_{p_n}\right)^{1/p_n}\right)\]
and there exists a solution to $(\ref{eqfa},\vec{a},1)$. \label{corsolved}
\end{enumerate}
In all those cases, the vacuum constraint equation \eqref{eqa} has a solution.
\end{cor}

\subsubsection{Connection with physics}
Actually, there is an astonishing connection with the physical origin of the constraint equation \eqref{eqa}.

\begin{thm}
Let $\vec{a} = (a_0, a_1, \ldots, a_p) \in \Rb^{p+1}$ and $\alpha_k, \Rb_k$ defined as in $(\ref{eqcontrainte1LLtsym}, \ref{defak})$. Then
\[\bar{f}_{\vec{a}} = a_0 \binom{n}{0} + a_1 \binom{n}{1} X + \ldots + a_p \binom{n}{p} X^p\]
is real-rooted if and only if the projected Lovelock polynomial can be factorised in a real product, ie. there exists $\nu \in \Rb^p$ such that:
\[\sum_{k=0}^p \alpha_k \Rb_k \propto \dfrac{1}{2^p}\delta_{a_1 b_1 a_2 b_2 \cdots a_p b_p}^{c_1 d_1 c_2 d_2 \cdots c_p d_p}
	\left(\Rr_{c_1 d_1}^{a_1 b_1} + \nu_1 \delta_{c_1 d_1}^{a_1 b_1} \right)\left(\Rr_{c_2 d_2}^{a_2 b_2} + \nu_2 \delta_{c_2 d_2}^{a_2 b_2}\right) \ldots \left(\Rr_{c_p d_p}^{a_p b_p} + \nu_p \delta_{c_p d_p}^{a_p b_p}\right). \]
\end{thm}
\begin{proof}
Let $\vec{a} = (a_0, a_1, \ldots, a_p) \in \Rb^{p+1}$. $\bar{f}_{\vec{a}} \in \Cb[X]$ can always be factorised in $\Cb$: let $\mu=(\mu_1, \ldots, \mu_p) \in \Cb^p$ such that
\[\bar{f}_{\vec{a}} = a_0 \binom{n}{0} + a_1 \binom{n}{1} X + \ldots + a_p \binom{n}{p} X^p = a_p \binom{n}{p} (X + \mu_1) \ldots (X + \mu_p).\]
Then for $0 \leq k \leq p$,
\[a_k \binom{n}{k} = a_p \binom{n}{p} \sigma_{p-k}(\mu),\]
ie.
\begin{equation}
a_k = a_p \dfrac{k! (n-k)!}{p! (n-p)!}\sigma_{p-k}(\mu).
\end{equation}
However, recall \eqref{defak}:
\begin{align*}
\alpha_k &= a_k\dfrac{(n-2k)!}{2^kk!(n-k)!} \\
&= a_p \dfrac{k! (n-k)!}{p! (n-p)!}\sigma_{p-k}(\mu)\dfrac{(n-2k)!}{2^kk!(n-k)!} \\
&= \dfrac{a_p}{p!(n-p)!}\sigma_{p-k}(\mu)\dfrac{(n-2k)!}{2^k} \\
&= \dfrac{a_p}{p!(n-p)!}\sigma_{p-k}(2\mu)\dfrac{(n-2k)!}{2^p}.
\end{align*}

%formules en sigmak 21 les voici :
%\begin{align*}
%\delta_{abcd}^{efgh}R_{ef}^{ab}R_{gh}^{cd} &= 2^2R^2 \\
%\delta_{abcd}^{efgh}R_{ef}^{ab}\delta_{gh}^{cd} &= 2^2R(n-2)(n-3) \\
%\delta_{abcd}^{efgh}\delta_{ef}^{ab}\delta_{gh}^{cd} &= 2^2n(n-1)(n-2)(n-3) \\
%\frac{1}{2^3}\delta_{......}^{......}(R_{..}^{..}+\alpha_1 \delta_{..}^{..})\cdot(R_{..}^{..}+\alpha_2 \delta_{..}^{..})\cdot(R_{..}^{..}+\alpha_3 \delta_{..}^{..}) &= \sigma_3(\alpha)n(n-1)(n-2)(n-3)(n-4)(n-5)R^0 \\
%&+ \sigma_2(\alpha)(n-2)(n-3)(n-4)(n-5)R^1 \\
%&+ \sigma_1(\alpha)(n-4)(n-5)R^2 \\
%&+ \sigma_0(\alpha)R^3 \\
%&= \binom{n}{6}6!\left[\frac{\sigma_3(\alpha)}{\binom{n}{0}}R^0 + \frac{\sigma_2(\alpha)}{\binom{n}{2}}R^1 + \frac{\sigma_1(\alpha)}{\binom{n}{4}}R^2 + \frac{\sigma_0(\alpha)}{\binom{n}{6}}R^3\right]
%\end{align*}

Meanwhile, let $\nu \in \Cb^p$. It can be shown by successive developments along rows that for $1 \leq k \leq p$,
\begin{align}
\dfrac{1}{2^p}\delta_{a_1 b_1 a_2 b_2 \cdots a_p b_p}^{c_1 d_1 c_2 d_2 \cdots c_p d_p}
	&\Rr_{c_1 d_1}^{a_1 b_1} \Rr_{c_2 d_2}^{a_2 b_2} \ldots \Rr_{c_k d_k}^{a_k b_k} \delta_{c_{k+1} d_{k+1}}^{a_{k+1} b_{k+1}} \ldots \delta_{c_p d_p}^{a_p b_p} \nonumber \\
	&= (n-2k)(n-2k-1)\ldots(n-2p+2)(n-2p+1)\Rb_k.
\end{align}
Hence, if we develop the following product, we get
\begin{align}
\dfrac{1}{2^p}\delta_{a_1 b_1 a_2 b_2 \cdots a_p b_p}^{c_1 d_1 c_2 d_2 \cdots c_p d_p}
	&\left(\Rr_{c_1 d_1}^{a_1 b_1} + \nu_1 \delta_{c_1 d_1}^{a_1 b_1} \right)\left(\Rr_{c_2 d_2}^{a_2 b_2} + \nu_2 \delta_{c_2 d_2}^{a_2 b_2}\right) \ldots \left(\Rr_{c_p d_p}^{a_p b_p} + \nu_p \delta_{c_p d_p}^{a_p b_p}\right) \nonumber \\
	&= \sum_{k=0}^p \sigma_{p-k}(\nu) (n-2k)(n-2k-1)\ldots(n-2p+2)(n-2p+1) \Rb_{k} \\
	&= \dfrac{1}{(n-2p)!} \sum_{k=0}^p \sigma_{p-k}(\nu) (n-2k)!~\Rb_{k}.
%\\	&= \sum_{k=0}^p \sigma_{k}(\nu) (n-2p+2k)(n-2p+2k-1)\ldots(n-2p+2)(n-2p+1) \Rb_{p-k} \\
%	&= \sum_{k=0}^p \sigma_k(\nu) (2k)! \binom{n-2p+2k}{2k} \Rb_{p-k} \\
%	&= \binom{n}{2p}(2p)! \sum_{k=0}^p \frac{\sigma_{p-k}(\nu)}{\binom{n}{2k}(2k)!}\Rb_{k} \\
%	&= \binom{n}{2p}(2p)! \sum_{k=0}^p \frac{\sigma_{k}(\nu)}{\binom{n}{2p-2k}(2p-2k)!}\Rb_{p-k}
\end{align}
So,
\begin{align}
f_{\vec{a}}(\lambda(\gamma))
	&= \sum_{k=0}^p a_k \sigma_k(\lambda(\gamma)) \\
%	&= \sum_{k=0}^p a_p \dfrac{\binom{n}{p}}{\binom{n}{k}} \sigma_{p-k}(\mu) \sigma_k(\lambda(\gamma)) \\
	&= \sum_{k=0}^p \alpha_k \Rb_k \\
	&= \frac{a_p}{2^p p!(n-p)!} \sum_{k=0}^p \sigma_{p-k}(2\mu) (n-2k)!~\Rb_k \\
	&= \frac{a_p (n-2p)!}{2^p p!(n-p)!} \times \dfrac{1}{2^p} \delta_{a_1 b_1 a_2 b_2 \cdots a_p b_p}^{c_1 d_1 c_2 d_2 \cdots c_p d_p}
\left(\Rr_{c_1 d_1}^{a_1 b_1} + 2\mu_1 \delta_{c_1 d_1}^{a_1 b_1} \right)\left(\Rr_{c_2 d_2}^{a_2 b_2} + 2\mu_2 \delta_{c_2 d_2}^{a_2 b_2}\right) \ldots \left(\Rr_{c_p d_p}^{a_p b_p} + 2\mu_p \delta_{c_p d_p}^{a_p b_p}\right).
\end{align}
Hence the roots of the factorised form of the projected Lovelock product are twice the roots of $\bar{f}_{\vec{a}}$.
\end{proof}

\begin{rem}
The factorised form of the Lovelock product was introduced in \cite{Zanelli91}, where it is seen as a product of concircular curvature tensors. They are tensors of the form
\[\Rr_{a b}^{c d} + \rho \delta_{a b}^{c d},\]
$\rho$ being a scale of curvature, ie. a scale of inverse squared length (see \cite{Banados90}, \cite{Yano70} and \cite{Yano84} for references about concircular curvature). Each $\Rr_{c_k d_k}^{a_k b_k} + 2\mu_k \delta_{c_k d_k}^{a_k b_k}$ has this shape, and the $\mu_k$'s represent $p$ different curvature scales. All the $\mu_k$'s are real and non-negative if and only if the associated length scales are real.

The theorem of Walsh \ref{thmLac17-3}\ref{thmWalsh} then implies that
\begin{equation}
\text{\parbox{60mm}{all the concircular length scales of the projected Lovelock polynomial are real}} \quad \Longrightarrow \quad \text{the constraint equation \eqref{eqa} have a solution}.
\end{equation}

It would be an interesting question to seek what it means, geometrically and physically, for a manifold to be a solution of the projected Lovelock equations with real concircular length scales.
\end{rem}

\begin{rem}
In the case where $n$ is odd, $p = p_n$ and all the roots are equal, with
\[\rho := 2\mu_1 = \ldots = 2\mu_{p_n},\]
we get
\begin{align*}
\sum_{k=0}^{p_n} \alpha_k \Rb_{k}
	&\propto \dfrac{1}{2^{p_n}}\delta_{a_1 b_1 a_2 b_2 \cdots a_{p_n} b_{p_n}}^{c_1 d_1 c_2 d_2 \cdots c_{p_n} d_{p_n}}
	\left(\Rr_{c_1 d_1}^{a_1 b_1} + \rho \delta_{c_1 d_1}^{a_1 b_1} \right)\left(\Rr_{c_2 d_2}^{a_2 b_2} + \rho \delta_{c_2 d_2}^{a_2 b_2}\right) \ldots \left(\Rr_{c_{p_n} d_{p_n}}^{a_{p_n} b_{p_n}} + \rho \delta_{c_{p_n} d_{p_n}}^{a_{p_n} b_{p_n}}\right) \\
	&\propto \Pf\left(\Omega^{a b} + \frac{1}{l^2}e^a e^b\right),
\end{align*}
where $\Omega^{a b}$ is the curvature 2-form, $e^a$ the vielbein 1-form of $(\Mc,\gamma)$ in the tetrad formalism, and $l^2 = \rho^{-1}$ the square of a length. This writing using the Pfaffian of a modified curvature form is called the Lagrangian of Born-Infeld gravity, by analogy with the electrodynamics Born-Infeld theory. It is physically and mathematically a source of interest (see \cite{Zanelli91}, \cite{Banados90}, \cite{Zanelli00}, \cite{Concha17}, \cite{Zanelli12}, \cite{Jimenez17}, \cite{Jana17}).

Here again, one must however be careful: the Born-Infeld Lagrangian usually lives in a $n+1$-dimensional Lorentzian space-time. Here we had projected the Lovelock field equations on a $n$-dimensional Riemannian space and only kept the Hamiltonian constraint. The situations are similar but not equal.
\end{rem}

\subsection{Analytic concavity}

\subsubsection{Concavifying functions}

Asking $\left(\sum_{k=0}^{p}a_k \sigma_k\right)^{1/p}$ to be concave is the strongest concavity hypothesis that can be done. Indeed, if $a_{p} \neq 0$, we have for $z \rightarrow +\infty$
\[\left(\sum_{k=0}^{p}a_k \sigma_k (z,\ldots,z) \right)^{\alpha} = \left(\sum_{k=0}^{p}a_k \binom{n}{k} z^k \right)^{\alpha} \sim \left(a_{p} \binom{n}{p}\right)^{\alpha} z^{\alpha p}\]
which cannot be concave if $\alpha > 1/p$. Roughly speaking: when one starts from $1/p$, the smaller $\alpha$ is, the most likely $\left(\sum_{k=0}^{p}a_k \sigma_k\right)^{\alpha}$ is concave.

To precise this point, let us look at a general map $P \in \Cc^2(\Gamma_n \rightarrow \Rb_+^*)$ and a strictly increasing function $F \in \Cc^2(\Rb_+^* \rightarrow \Rb)$. Then
\begin{equation}
H\left(F \circ P\right) = \dfrac{F' \circ P}{P} \left[P H(P) + \left(\dfrac{F'' \circ P}{F'\circ P}P\right) D(P) \otimes D(P)\right]. \label{HFP}
\end{equation}
$F\circ P$ is concave on $\Gamma_n$ if and only if $H\left(F \circ P\right)$ is everywhere negative.

\begin{defi}
$F$ is said to be a \emph{concavifying} function of $P$ if and only if $F \circ P$ is concave.
\end{defi}

For instance, exponentiating to a small number amounts to take $F(u)=u^\alpha$, $\alpha > 0$, and then
\begin{equation}
H\left(P^\alpha\right) = \alpha P^{\alpha-2} \left[P H(P) + (\alpha -1) D(P) \otimes D(P)\right]. \label{HPalpha}
\end{equation}
$D(P) \otimes D(P)$ being positive, $u\mapsto u^\alpha$ is more likely to concavify $P$ if $\alpha$ is small. The $\alpha \rightarrow 0$ limit corresponds to $\ln P$. But that is only a particular way to concavify a function. For our purpose, all concavifying function is suitable.

\begin{lem}
Let $(a_0, a_1, \ldots, a_p) \in \Rb^{p+1}$. Let us assume that there exists some function $F : \Rb_+^* \rightarrow \Rb$ such that
\begin{enumerate}[(i)]
\item $F \in \Cc^2(\Rb_+^* \rightarrow \Rb)$; \label{Fi}
\item $F$ is strictly increasing; \label{Fii}
\item $\lim\limits_{t \rightarrow +\infty}F(t) = +\infty$; \label{Fiii}
\item $\lim\limits_{t \rightarrow a_0}F(t) = 0$; \label{Fiv}
\item $x \in \Gamma_n \longmapsto F \left(a_0 + a_1 \sigma_1(x) + \ldots + a_p \sigma_p(x)\right)$ is concave. \label{Fv}
\end{enumerate}
Then the Theorem \ref{thmLi} applies to
\[(\Gamma, f) = \left(\Gamma_n, F \circ \left(a_0 + a_1 \sigma_1 + \ldots + a_p \sigma_p\right)\right).\]
\end{lem}
\begin{proof}
It is easy to check that all the hypothesis of Theorem \ref{thmLi} are satisfied.
\end{proof}

\begin{rem}
In particular, if $(0, a_1, a_2, \ldots, a_{p_n})$ are such that some function $F$ satisfying $(\ref{Fi}-\ref{Fv})$ exists, then there exists a solution to \eqref{eqa}.
\end{rem}

Looking at \eqref{HPalpha}, one could think that it would suffice to take $\alpha \ll 0$ to get eventually $H\left(P^\alpha\right)$ negative. However that is not a good way to proceed. Indeed, being ``too'' concave prevents to satisfy condition $(\ref{Fiii})$:
\begin{prop}
Let $P \in \Cc^2(\Gamma_n \rightarrow \Rb_+^*)$ with $P(\infty) = \infty$ and a strictly increasing function $F \in \Cc^2(\Rb_+^* \rightarrow \Rb)$ such that
\[H(F \circ P) \equiv PH(P) + \left(\dfrac{F'' \circ P}{F'\circ P}P\right) D(P) \otimes D(P)\]
is negative. If there exists some $\tau > 1$ such that uniformly
\begin{equation}
\dfrac{F'' \circ P}{F'\circ P}P \leq -\tau, \label{majtau}
\end{equation}
then $F$ is bounded from above.
\end{prop}
\begin{proof}
We just have to integrate twice.
\begin{align*}
\eqref{majtau} &\Longrightarrow \dfrac{F''(t)}{F'(t)}t \leq -\tau \qquad \forall~t \in \Rb_+^* \\
&\Longleftrightarrow \dfrac{F''(t)}{F'(t)} \leq - \dfrac{\tau}{t} \qquad \forall~t \in \Rb_+^* \\
&\Longleftrightarrow \ln(F'(t)) - \ln(F'(1)) \leq - \tau \ln t \qquad \forall~t \in \Rb_+^* \\
&\Longleftrightarrow F'(t) \leq \dfrac{F'(1)}{t^\tau} \qquad \forall~t \in \Rb_+^* \\
&\Longleftrightarrow F(t) \leq F(1) + \dfrac{F'(1)}{\tau-1}\left(1- \dfrac{1}{t^{\tau-1}}\right) \qquad \forall~t \in \Rb_+^* \\
&\Longrightarrow F(t) \leq F(1) + \dfrac{F'(1)}{\tau-1} \qquad \forall~t \in \Rb_+^* \\
\end{align*}
\end{proof}

%%%%%%%%%%%%%%%%%%%%%%%%%%%%%%%%%%%%%%%%%%%%%%%%%%%%%%%%%%%%%%%%%%%%%%%%%%%%%%%%%%%%%%%%%%%%%
% Ici on peut insérer le calcul de la tau concavité pour p=25, brouillons sigma k 24 et 25. %
%%%%%%%%%%%%%%%%%%%%%%%%%%%%%%%%%%%%%%%%%%%%%%%%%%%%%%%%%%%%%%%%%%%%%%%%%%%%%%%%%%%%%%%%%%%%%

Hence a function $F$ that is concavifying but without preventing to diverge at infinity must have $-F''(t)t/F'(t)$ close to $1$ at some place. If a logarithm is not enough, it is natural to think of successive iterations of logarithms.
\begin{defi}
For $k \in \Nb$, we write
\begin{align}
\ln[0] &= \Id, \nonumber \\
\ln[k] &= \overbrace{\ln \circ \ldots \circ \ln}^k. \label{lnk}
\end{align}
\end{defi}
\begin{prop}
For $x \in \Rb$ large enough, one can show by induction that:
\begin{align}
\ln[k]' &= \dfrac{1}{\ln[k-1] \times \ldots \times \ln[2] \times \ln[1] \times \ln[0]}, \\
\ln[k]'' &= -\dfrac{1+\ln[k-1](1+\ldots(1+\ln[3](1+\ln[2](1+\ln[1])))\ldots)}{\left(\ln[k-1] \times \ldots \times \ln[2] \times \ln[1] \times \ln[0]\right)^2}, \\
-\dfrac{\ln[k]''(x)}{\ln[k]'(x)}x &= 1 + \dfrac{1+\ln[k-1](1+\ldots(1+\ln[3](1+\ln[2]))\ldots)}{\ln[k-1] \times \ldots \times \ln[2] \times \ln[1]}(x).
\end{align}
\end{prop}
Up to positive constants to add so that all the functions are well-defined, this might be a strong concavifying family of functions. Actually, we used one of them to entirely solve the $p=2$ case.

\subsubsection{\texorpdfstring{$p=2$}{p=2}}
Let $\vec{a} = (a_0, a_1, a_2) \in \left(\Rb_+\right)^3$. If $(a_0, a_1, a_2) \notin \Kc_n^2$, we can not use the Corollary \ref{corsolve} to solve $(\ref{eqfa}, \vec{a}, 1)$. Though, there exists in all cases a concavifying function of $f_{\vec{a}}$.
\begin{thm} \label{thmp2}
Let $a, b, c \in \Rb_+$, $s \in \Rb$, such that $a>0$ and $s + \ln a > 0$. Then
\[x \in \Gamma_n \longmapsto \ln(s + \ln(a + b\sigma_1(x) + c\sigma_2(x))) \quad \text{ is concave}.\]
\end{thm}
\begin{proof}
If $c=0$, $a + b\sigma_1$ is an affine function, thus concave. \eqref{HFP} shows that the logarithm of a concave function is concave as well, so this property is inherited by $\ln \circ (a + b \sigma_1)$ and later $\ln \circ (s + \ln \circ (a + b\sigma_1))$.

Let us suppose that $c>0$. Let
\begin{align*}
g &: x \in \Gamma_n \longmapsto a + b\sigma_1(x) + c\sigma_2(x), \\
F &: t \in \ ]a, +\infty[\ \longmapsto \ln(s+\ln t), \\
\tau &: t \in \ ]a, +\infty[\ \longmapsto -\dfrac{F''(t)}{F'(t)}t, \\
f &:= F \circ g.
\end{align*}
Then
\begin{align*}
H(f) &= \dfrac{F' \circ g}{g} \left[g H(g) + \left(\dfrac{F'' \circ g}{F' \circ g} g \right) D(g) \otimes D(g)\right] \\
&\equiv g H(g) - \left(\tau \circ g\right) D(g) \otimes D(g).
\end{align*}
Let us compute the determinant of this matrix. For $x \in \Gamma_n$, we have:
\begin{align*}
H	&= \left(\begin{array}{c|c|c} & & \\ H_1 & \ldots & H_n \\ & & \end{array}\right)
	:= H(g)(x) = c \begin{pmatrix} 0 & 1 & \ldots & 1 \\ 1 & 0 & & \vdots \\ \vdots & & \ddots & 1 \\ 1 & \ldots & 1 & 0
\end{pmatrix}
	= c \left(\Ib \otimes \Ib - \Id\right), \\
D	&:= D(g)(x) = \begin{pmatrix} g_1(x) \\ \vdots \\ g_n(x) \end{pmatrix}
	= \begin{pmatrix} b + c (\sigma_1(x) - x_1) \\ \vdots \\ b + c(\sigma_1(x) - x_n) \end{pmatrix}
	= (b + c\sigma_1(x)) \Ib - cx, \\
F'(t)	&:= \dfrac{1/t}{s + \ln t}, \\
F''(t)	&= - \dfrac{1}{t^2}\times \dfrac{s + \ln t + 1}{(s + \ln t)^2}, \\
\tau(t)	&:= -\dfrac{F''(t)}{F'(t)}t = 1 + \dfrac{1}{s + \ln t}.
\end{align*}
Thus,
\begin{align*}
\det H(f)(x) &\equiv \det \left[g(x) H(g)(x) - \tau(g(x)) D(g)(x) \otimes D(g)(x) \right] \\
	&= \det \left[g(x) H_1 - \tau(g(x)) g_1(x)D \ | \ \ldots \ | \ g(x)H_n - \tau(g(x)) g_n(x)D \right].
\end{align*}
Using the multilinearity of the determinant, we have
\[\det H(f)(x) \equiv g(x)^n \det H - \tau(g(x))g(x)^{n-1} \sum_{j=1}^n\det\left[H_1 \ | \ \ldots \ | \ g_j(x)D \ | \ \ldots \ | \ H_n\right].\]
We expand along the $j$-th column, and we denote by $\bar{H}_{ij}$ the minors of $H$. Then we get
\begin{align*}
\sum_{j=1}^n \det\left[H_1\ | \ \ldots \ | \ g_j(x)D \ | \ \ldots \ | \ H_n\right] &= \sum_{j=1}^n g_j(x) \sum_{i=1}^n (-1)^{i+j} g_i(x) \bar{H}_{ij}, \\
&= ^t\!\!D \com H D \\
&= ^t\!\!D ^t\!\com H D \\
&= ^t\!\!D (\det H) H^{-1} D.
\end{align*}
So we have to compute $\det H$ and $H^{-1}$. It is a classical result that $H$ can be diagonalised in an orthonormal basis into
\[c\begin{pmatrix} n-1 & 0 & \ldots & 0 \\ 0 & -1 &  & \vdots \\ \vdots &  & \ddots & 0 \\ 0 & \ldots & 0 & -1 \end{pmatrix},\]
hence $\det H = (n-1)(-1)^{n-1}c^n$. Moreover, one can check that
\[H^{-1} = \dfrac{1}{(n-1)c}\left(\Ib \otimes \Ib - (n-1)\Id\right).\]
Hence,
\begin{align*}
^t\!D (\det H) H^{-1} D
	&= ^t\!\!D\left[(-c)^{n-1}\left(\Ib^t\Ib - (n-1)\Id\right)\right]D \\
	&= (-c)^{n-1}\left[\left(^t\Ib D\right)^2 - (n-1)^tDD\right].
\end{align*}
Let us determine each term:
\begin{align*}
\left(^t\Ib D\right)^2
	&= \left(nb + (n-1)c\sigma_1(x)\right)^2 \\
	&= n^2b^2 + 2n(n-1)bc\sigma_1(x) + (n-1)^2c^2\sigma_1(x)^2, \\
^tDD
	&= (b + c\sigma_1(x))^2~^t\Ib\Ib - 2(b + c\sigma_1(x))c ^t\Ib x + c^2~^txx \\
	&= nb^2 + 2(n-1)bc\sigma_1(x) + (n-2)c^2\sigma_1(x)^2 + c^2\left(\sigma_1(x)^2 - 2\sigma_2(x)\right),
\end{align*}
mix it all together:
\begin{align*}
^t\!D (\det H) H^{-1} D
	&= (-c)^{n-1}\left[nb^2 + 2(n-1)bc\sigma_1(x) + 2(n-1)c^2 \sigma_2(x)\right] \\
	&= (-c)^{n-1}\left[nb^2 + 2(n-1)c (g(x) - a)\right],
\end{align*}
and finally,
\begin{align*}
\det H(f)(x)
	&\equiv \det \left[g(x) H(g)(x) - \tau(g(x)) D(g)(x) \otimes D(g)(x) \right] \\
	&\equiv g(x) \det H - \tau(g(x)) ^t\!D (\det H) H^{-1} D \\
	&= (-c)^{n-1}\left[(n-1)c g(x) - \tau(g(x)) \left(nb^2 + 2(n-1)c(g(x)-a)\right)\right] \\
	&= (-c)^{n-1}\left[(n-1)c g(x)\left(1 -2\tau(g(x))\right) - \tau(g(x)) \left(nb^2 - 2(n-1)ac\right)\right].
\end{align*}
The last expression is not insignificant: it represents the discriminant of the restriction of $g$ to $\Rb_+^*\cdot \Ib$ seen as a polynomial, namely
\[\bar{g} = g(X, \ldots, X) = a + nb X + \dfrac{n(n-1)}{2}c X^2.\]
For this reason we set
\[\Delta_{\bar{g}} := nb^2 - 2(n-1)ac,\]
and we can write
\begin{align*}
\det H(f)(x) &\equiv (-c)^{n-1}\left[(n-1)cg(x)\left(-1-\dfrac{2}{s + \ln(g(x))}\right) - \left(1 + \dfrac{1}{s + \ln(g(x))}\right)\Delta_{\bar{g}}\right] \\
&\equiv -(-c)^{n-1}\left[(n-1)cg(x)\left(s + 2 + \ln(g(x))\right) + \left(s + 1 + \ln(g(x))\right)\Delta_{\bar{g}}\right] \\
&= -(-c)^{n-1}\varphi(g(x)),
\end{align*}
with
\[\varphi : t \in \ ]a, +\infty[ \ \longmapsto (n-1)ct(s + 2 + \ln t) + (s + 1 + \ln t)\Delta_{\bar{g}}.\]
Now we intend to show that $\varphi > 0$. This will prove that the sign of $\det H(f)$ is constant and equals $(-1)^n$, so that the concavity of $f$ is constant on $\Gamma_n$.

Let us assume that $\Delta_{\bar{g}} \geq 0$. We have taken $s$ such that $s + \ln a > 0$, so $\varphi > 0$ on $]a, +\infty[$.

Let us assume that $\Delta_{\bar{g}} < 0$. We are going to determine the sign of the minimum of $\varphi$. Let $t \in \ ]a, +\infty[$.
\begin{align*}
\varphi'(t) &= (n-1)c(s + 3 + \ln t) + \dfrac{\Delta_{\bar{g}}}{t}, \\
\varphi'(t) &> 0 \quad \Longleftrightarrow \quad t(s + 3 + \ln t) > -\dfrac{\Delta_{\bar{g}}}{(n-1)c} > 0.
\end{align*}
Yet $t \mapsto t(s + 3 + \ln t)$ is a strictly increasing map, with
\begin{align*}
\lim\limits_{t \rightarrow 0} t(s + 3 + \ln t) &= 0, \\
\lim\limits_{t \rightarrow +\infty} t(s + 3 + \ln t) &= +\infty,
\end{align*}
so there exists some unique $t^* \in \ ]0, +\infty[$ such that
\[t^*(s + 3 + \ln t^*) = -\dfrac{\Delta_{\bar{g}}}{(n-1)c}.\]
On the right hand,
\[-\dfrac{\Delta_{\bar{g}}}{(n-1)c} = 2a - \dfrac{nb^2}{(n-1)c} < 2a.\]
On the left hand, in $t = a$,
\[a(s + 3 + \ln a) > 3a.\]
Hence $t^* < a$. Then $\varphi'(t) > 0$ for all $t \in \ ]a, +\infty[$, and $\varphi$ reaches a global minimum in $t=a$. Let us determine the value of this minimum.
\begin{align*}
\varphi(a) &= (n-1)ac(s + 2 + \ln a) + (s + 1 + \ln a)\Delta_{\bar{g}} \\
&> 2(n-1)ac + \Delta_{\bar{g}} = b^2 \\
&\geq 0.
\end{align*}
So $\varphi > 0$ on $]a, +\infty[$, and $\det H(f)$ has a constant sign on $\Gamma_n$. This implies that the eigenvalues of $H(f)$ never vanish on $\Gamma_n$, which is a connected set. Hence the concavity of $f$ is constant on $\Gamma_n$: if one is able to show that $f$ is concave in one point, then it is concave everywhere. We shall do it on the diagonal.

We have to recall \eqref{HFP}:
\[H(f) \equiv gH(g) - \tau(g)D(g)\otimes D(g).\]
Let $u \in \Rb^n\backslash \{0\}$, $z \in \Rb_+^*$. We have
\begin{align*}
^tuH(f)(z\cdot\Ib)u &\equiv g(z\cdot\Ib) ^tuH(g)(z\cdot\Ib)u - \tau(g(z\cdot\Ib))\left(^tu D(g)(z\cdot\Ib)\right)^2 \\
&= \left(a + nbz + \dfrac{n(n-1)}{2}cz^2\right)~^tu c \left(\Ib \otimes \Ib - \Id\right)u \\
&\qquad - \left(1 + \dfrac{1}{s + \ln \left(a + nbz + \frac{n(n-1)}{2}cz^2\right)}\right)\left(^tu \left[(b + c(n-1)z)\Ib\right]\right)^2 \\
&\underset{z\rightarrow+\infty}{\sim} c^2 z^2 \left(\dfrac{n(n-1)}{2}\left[(u|\Ib)^2 - (u|u)\right] - (n-1)^2(u|\Ib)^2\right) \\
&= c^2 z^2 \dfrac{(n-1)}{2}\left(-(n-2)(u|\Ib)^2 - n(u|u)\right) \\
&< 0 \quad \text{for } n \geq 2.
\end{align*}
Hence $H(f)(z\cdot\Ib)$ is a negative matrix for sufficiently large $z$, thus $f$ is concave everywhere.
\end{proof}

\begin{cor}
For all $\vec{a} = (a_0, a_1, a_2) \in \left(\Rb_+\right)^3$, there exists a solution to $(\ref{eqfa}, \vec{a}, 1)$.
\end{cor}
\begin{proof}
The cases $a_0=0$ and $a_2=0$ were already handled in the Corollary \ref{corsolve}\ref{corsolvea} and \ref{corsolve}\ref{corsolveb}.

Let us assume that $a_0 > 0$. We set $s = 1 - \ln a_0$. Then according to the Theorem \ref{thmp2}, $\ln \circ (s + \ln \circ (a_0 + a_1 \sigma_1 + a_2 \sigma_2))$ is concave. So the Theorem \ref{thmLi} applies to
\[(\Gamma,f) = \left((\Gamma_n, \ln \circ (s + \ln \circ (a_0 + a_1 \sigma_1 + a_2 \sigma_2))\right).\]
\end{proof}

\begin{rem}
We established the concavity of $f$ by showing it on the diagonal $\Rb_+^*\cdot\Ib$, after having proved that it was sufficient. Actually, the example of a single $\sigma_k$ leads to think that it is always so: $\sigma_k^{\alpha}$ is concave on $\Gamma_k$ if and only if it is concave on the diagonal, ie. for $\alpha < 1/k$. See \cite{Lac17-3} for more details. This induces us to formulate the following conjecture.
\end{rem}

\begin{conj}
Let $\vec{a} = (a_0, a_1, \ldots, a_p) \in \left(\Rb_+\right)^{p+1}$, and $F \in \Cc^2(\Rb_+^* \rightarrow \Rb)$. Set
\[f := F \circ \left(a_0 + a_1 \sigma_1 + \ldots + a_p \sigma_p\right).\]
Then
\[f \text{ is concave on } \Gamma_n \quad \Longleftrightarrow \quad f \text{ is concave on } \Delta_+^*.\]
\end{conj}

It seems possible to always find concavifying enough functions, such as the family \eqref{lnk}, hence to generalise the Theorem \ref{thmp2}, but we were not able to prove it.
\begin{conj}
For all $\vec{a} = (a_0, a_1, \ldots, a_p) \in \left(\Rb_+\right)^{p+1}$, there exists $F \in \Cc^2(\Rb_+^* \rightarrow \Rb)$ such that
\[F \circ f_{\vec{a}} \text{ is concave on } \Gamma_n.\]
\end{conj}

\section{Implicit function theorem} \label{sectiontfi}
In this section, we shall not assume anymore that the coefficients $a_k$ are non-negative. We shall suppose that $a_k$, $k \geq 2$ are small with respect to $1$: this is the GR limit of Lovelock theories. More precisely, we intend to apply a local implicit function theorem. For that purpose, we consider the constraint equation in Sobolev spaces.

In connection with the Yamabe and $\sigma_k$-Yamabe problems and the conformal method for the constraint equations, we look for a solution of \eqref{eqfa} in a given conformal class.

Remember that
\begin{equation}
2^* = \frac{2n}{n-2}, \qquad 2^* -1 = \frac{n+2}{n-2}, \qquad 2^* - 2 = \frac{4}{n-2}, \qquad 2^* + 2 = \frac{4(n-1)}{n-2},
\end{equation}
and that
\begin{equation} \label{SchoutenRicci}
S[\gamma]_{ij} = \dfrac{1}{n-2}\left[\Rr[\gamma]_{ij} - \dfrac{\Rr[\gamma]}{2(n-1)}\gamma_{ij}\right], \qquad S[\gamma] = \dfrac{\Rr[\gamma]}{2(n-1)}.
\end{equation}

The following formulas come from the classical conformal calculus (see e.g. \cite{STW12}):
\begin{lem}
Let $u \in \Cc^\infty(\Mc)$, $u > 0$, $\tilde{\gamma} = u^{\frac{4}{n-2}}\gamma$. Then
\begin{equation}
S[\tilde{\gamma}]_{ij} = \dfrac{2}{(n-2)u}\left(-\nabla[\gamma]_i \nabla[\gamma]_j u + \dfrac{n}{n-2} \dfrac{\nabla[\gamma]_i u \nabla[\gamma]_j u}{u} -  \dfrac{\nabla[\gamma]_k u \nabla[\gamma]_l u \gamma^{k l}}{u}\dfrac{\gamma_{ij}}{n-2}\right) + S[\gamma]_{ij}, \label{Sijtilde}
\end{equation}
so
\begin{align*}
\sigma_1\left(\lambda\left(\tilde{\gamma}\right)\right) &= S[\tilde{\gamma}] \\
&= S[\tilde{\gamma}]_{ij} \tilde{\gamma}^{ij} \\
&= S[\tilde{\gamma}]_{ij} u^{-2^*+2} \gamma^{ij} \\
&= \dfrac{2 u^{-2^*+1}}{n-2}\left(-\Delta[\gamma] u\right) + S[\gamma]u^{-2^*+2},
\end{align*}
and
\begin{equation}
S[\tilde{\gamma}]u^{2^*-1} = -\dfrac{2}{n-2} \Delta[\gamma] u + S[\gamma] u. \label{Sgammatilde}
\end{equation}
In terms of scalar curvature, that gives
\begin{equation}
\Rr\left[u^{2^*-2}\gamma\right] u^{2^*-1} = -\left(2^* + 2\right)\Delta[\gamma] u + \Rr[\gamma] u.
\end{equation}
\end{lem}

\begin{cor}
Let $\varepsilon > 0$. Then, for $a_0, \ldots, a_p \in \Rb^{p+1}$, $u : \Mc \rightarrow \Rb_+^*$,
\[u \in W^{2+\varepsilon,2}(\Mc) \qquad \Longrightarrow \qquad \sum_{k=0}^p a_k \sigma_k\left(\lambda\left(u^{\frac{4}{n-2}}\gamma\right)\right) \in W^{\varepsilon,2}(\Mc).\]
\end{cor}
\begin{proof}
Let $u \in W^{2+\varepsilon,2}(\Mc)$. Because of \eqref{Sijtilde}, $S[\tilde{\gamma}]_{ij} \in W^{\varepsilon,2}\left(\Mc \rightarrow \Sc_n(\Rb)\right)$.

Let $0 \leq k \leq p$. For all $M \in \Sc_n(\Rb)$, if $\lambda(M)$ denotes the eigenvalues of $M$, then $\sigma_k(\lambda(M))$ coincides to the $k$-th coefficient of the characteristic polynomial of $M$, starting from the highest degree. So $\sigma_k(\lambda(M))$ is a polynomial in the coefficients of $M$, and
\[\left(\begin{array}{ccc}
\Sc_n(\Rb) & \longrightarrow & \Rb \\
M & \longmapsto & \sigma_k(\lambda(M)) \end{array} \right) \quad \in \quad \Cc^\infty(\Sc_n(\Rb) \rightarrow \Rb).\]
By composition, we deduce the regularity of $\sum_k a_k \sigma_k$.
\end{proof}

Thus, we can define the following map.
\begin{defi}
Let $\varepsilon > 0$. We introduce:
\[T[\gamma] : \begin{array}{ccc}
W^{2+\varepsilon,2}(\Mc) \times \Rb^{p+1} & \longrightarrow & W^{\varepsilon,2}(\Mc) \\
(u, \vec{a}) & \longmapsto & \displaystyle{\sum_{k=0}^p a_k \sigma_k\left(\lambda\left(u^{\frac{4}{n-2}}\gamma\right)\right)}.
\end{array}\]
\end{defi}

Now we are ready to formulate the main result of this section:
\begin{thm} \label{thminv}
Let us assume that $\Rr[\gamma] < 0$.

For all $a_0, a_1 \in \Rb$, $a_1 \neq 0$, we set $\vec{\alpha} = (a_0, a_1, 0, \ldots, 0)$. We take $u = 1$. Then $T[\gamma]$ is $\Cc^1$ and
\[\derp{T[\gamma]}{u} (1,\vec{\alpha}) \qquad \text{ is invertible}.\]
\end{thm}
\begin{proof}
The proof is inspired from the Yamabe problem.

Let $u \in W^{2+\varepsilon,2}(\Mc)$, $u > 0$, and let $a_0, a_1 \in \Rb$, $a_1 \neq 0$. For example, let us take $a_1 > 0$. Let $h \in W^{2+\varepsilon,2}(\Mc)$. Then
\begin{align*}
\derp{T[\gamma]}{u}(u,a_0,a_1,0,\ldots,0)(h) &= \derp{}{t}\Big|_{t=0} T[\gamma](u+th,a_0,a_1,0,\ldots,0) \\
&= \derp{}{t}\Big|_{t=0} \left(a_0 + a_1 S\left[(u+th)^{\frac{4}{n-2}}\gamma\right] + 0 + \ldots + 0 \right) \\
&= a_1 \derp{}{t}\Big|_{t=0} \left(\dfrac{-\frac{2}{n-2}\Delta[\gamma] (u+th) + S[\gamma] (u+th)}{(u+th)^{2^* - 1}} \right) \\
&\equiv \bigg[\derp{}{t}\left(-\frac{2}{n-2}\Delta[\gamma] (u+th) + S[\gamma] (u+th)\right)(u+th)^{2^* - 1} \\
&\quad - \left(-\frac{2}{n-2}\Delta[\gamma] (u+th) + S[\gamma] (u+th)\right)\derp{}{t}(u+th)^{2^* - 1}\bigg]\bigg|_{t=0} \\
&= \left(-\frac{2}{n-2}\Delta[\gamma] h + S[\gamma] h\right) u^{2^*-1} - \left(-\frac{2}{n-2}\Delta[\gamma] u + S[\gamma] u\right)\left(2^*-1\right)hu^{2^*-2} \\
&\equiv \left(-\frac{2}{n-2}\Delta[\gamma] h + S[\gamma] h\right) u - \left(-\frac{2}{n-2}\Delta[\gamma] u + S[\gamma] u\right)\left(2^*-1\right)h.
\end{align*}
Now we take $u$ constant to 1. We obtain
\begin{align*}
\derp{T[\gamma]}{u}(1,\vec{\alpha})(h) &\equiv -\frac{2}{n-2}\Delta[\gamma] h + S[\gamma] h - \left(2^*-1\right) S[\gamma] h \\
&= -\frac{2}{n-2}\Delta[\gamma] h - \left(2^*-2\right) S[\gamma] h \\
&\equiv -\left(\Delta[\gamma] h + 2 S[\gamma] h \right).
\end{align*}

We fix the metric $\gamma$. We set
\begin{equation}
c := -2 S[\gamma].
\end{equation}
$\Rr[\gamma] < 0$ by hypothesis, so $c > 0$ according to \eqref{SchoutenRicci}. Then we define
\begin{equation}
L : \begin{array}{ccc} W^{2+\varepsilon,2}(\Mc) & \longrightarrow & W^{\varepsilon,2}(\Mc) \\
h & \longmapsto & - \Delta h + c h,
\end{array}
\end{equation}
and we shall prove that this operator is invertible.

\begin{itemize}
\item \textbf{Injectivity:} Let $h \in W^{2+\varepsilon,2}(\Mc)$ be such that $L(h) = 0$. Then
\begin{align*}
0 &= \int_\Mc hL(h) \d v \\
&= \int_\Mc \left(\left|\nabla h\right|^2 + c h^2 \right) \d v.
\end{align*}
$c > 0$, so
\[\left|\nabla h\right|^2 + c h^2 = 0\]
everywhere on $\Mc$, and $h=0$.

\item \textbf{Surjectivity:} Let $w \in W^{\varepsilon,2}(\Mc)$. We define
\begin{align}
I &: \begin{array}{ccc} W^{1,2}(\Mc) & \longrightarrow & \Rb \\ h & \longmapsto & \displaystyle{\int_\Mc \left(\left|\nabla h\right|^2 + c h^2\right) \d v} \end{array}, \\
\nu &:= \inf_{\begin{array}{c} h \in W^{1,2}(\Mc) \\ \int w h \neq 0 \end{array}} \dfrac{\displaystyle{I(h)}}{\displaystyle{\left(\int_\Mc w h \d v\right)^2}}
= \inf_{\begin{array}{c} h \in W^{1,2}(\Mc) \\ \int w h = 1 \end{array}} I(h).
\end{align}
%%%%%%%%%%%%%%%%%%%%%%%%%%%%%%% ça commence là.
There exists a minimising sequence $(h_n)$ in $W^{1,2}(\Mc)$ such that
\[\int_\Mc w h_n \d v = 1 \qquad \text{and} \qquad \int_\Mc \left(\left|\nabla h_n\right|^2 + c h_n^2\right) \d v \xrightarrow[n \rightarrow \infty]{} \nu.\]
Using a standard reasoning from the Yamabe theory, we can show that there exists some $h_\infty \in W^{1,2}(\Mc)$ such that
\begin{equation}
\int_\Mc w h_\infty \d v = 1 \qquad \text{and} \qquad I(h_\infty) = \nu. \label{whinf}
\end{equation}
In particular, $h_\infty \neq 0$ and $\nu > 0$.

Now we position on $h_\infty$ and look at the variations of $I$: let $\varphi \in \Cc^\infty(\Mc) \subset W^{1,2}(\Mc)$. Because $h_\infty$ minimises $I(h) / \left(\int wh\right)^2$ and because of \eqref{whinf}, we get
\begin{align*}
0 &= \derp{}{t}\bigg|_{t=0} \dfrac{I(h_\infty + t\varphi)}{\left(\int_\Mc w(h_\infty + t\varphi)\d v\right)^2} \\
&= \derp{}{t}\bigg|_{t=0} \dfrac{\displaystyle{\int_\Mc \left(\left|\nabla h_\infty + t \nabla \varphi\right|^2 + c\left|h_\infty + t \varphi\right|^2\right) \d v}}{\displaystyle{\left(\int_\Mc w\left(h_\infty + t\varphi\right)\d v\right)^2}} \\
&= \dfrac{\displaystyle{\left(\int_\Mc 2\left[\nabla_i h_\infty \nabla^i \varphi + c h_\infty \varphi \right] \d v\right) \left(\int_\Mc w h_\infty \d v\right)}}{\displaystyle{\left(\int_\Mc w h_\infty \d v\right)^4}} \\
& \qquad - \dfrac{\displaystyle{\left(\int_\Mc \left[\left|\nabla h_\infty\right|^2 + c\left|h_\infty \right|^2\right] \d v\right)2 \left(\int_\Mc w h_\infty \d v\right)\left(\int_\Mc w \varphi \d v\right)}}{\displaystyle{\left(\int_\Mc w h_\infty \d v\right)^4}} \\
&= 2\int_\Mc \left[\nabla_i h_\infty \nabla^i \varphi + c h_\infty \varphi \right] \d v - 2I(h_\infty)\int_\Mc w \varphi \d v \\
&= 2\int_\Mc \left[\nabla_i h_\infty \nabla^i \varphi + c h_\infty \varphi \right] \d v - 2\nu\int_\Mc w \varphi \d v.
\end{align*}
We set
\[\tilde{h}_\infty := \dfrac{h_\infty}{\nu}.\]
Then, for all $\varphi \in \Cc^\infty(\Mc)$,
\[\int_\Mc \left[\nabla_i \tilde{h}_\infty \nabla^i \varphi + c \tilde{h}_\infty \varphi \right] \d v = \int_\Mc w \varphi \d v.\]
That is to say,
\[-\Delta \tilde{h}_\infty + c \tilde{h}_\infty = w \qquad \text{in the sense of Aubin}.\]

We can apply the following result of T. Aubin:
\begin{thm}[Aubin]
Let $c > 0$, $\eta \geq 0$, $q \geq 1$.

Let $h, w$ be such that
\[-\Delta h + c h = w \qquad \text{in the sense of Aubin}.\]
Then
\[w \in W^{\eta,q} \Longrightarrow h \in W^{2+\eta,q}.\]
\end{thm}

Hence $\tilde{h}_\infty \in W^{2+\varepsilon,2}(\Mc)$ and
\[L(\tilde{h}_\infty) = -\Delta \tilde{h}_\infty + c \tilde{h}_\infty = w \qquad \text{on } \Mc,\]
so $L$ is surjective.
\end{itemize}

We have just proved that $L = \derp{T[\gamma]}{u}(1,\vec{\alpha})$ is invertible.
\end{proof}

\begin{rem}
The negativity of $\Rr[\gamma]$ is necessary for $\derp{T[\gamma]}{u}(1,\vec{\alpha})$ to be invertible: if $\Rr[\gamma] \geq 0$, we lose the injectivity and $T[\gamma]$ might have a kernel.
\end{rem}

\begin{thm} \label{tfi}
Let $a_0, a_1 \in \Rb$ such that $a_1(a_0 -1) > 0$.

Let $\varepsilon > 0$ and $\gamma$ be such that $\Rr[\gamma] < 0$. We set
\[\vec{\alpha} = (a_0, a_1, 0, \ldots, 0) \in \Rb^{p+1}.\]
Then there exists $\eta > 0$ such that for all $\vec{a} = (a_0, a_1, a_2, \ldots, a_p) \in \Rb^{p+1}$,
\[\left\|\vec{a} - \vec{\alpha}\right\| < \eta \qquad \Longrightarrow \qquad
\text{there exists $u_{\vec{a}} \in W^{2+\varepsilon,2}(\Mc)$ such that} \quad
\displaystyle{\sum_{k=0}^p a_k \sigma_k\left(\lambda\left(u_{\vec{a}}^{\frac{4}{n-2}}\gamma\right)\right) = 1}.\]
Moreover, $\vec{a} \longmapsto u_{\vec{a}}$ is $\Cc^1$.
\end{thm}
\begin{proof}
Let $a_0, a_1 \in \Rb$ with $a_1(a_0 -1) > 0$, and $\varepsilon > 0$.

We suppose that $\Rr[\gamma] < 0$. According to Yamabe theorem, there exists $v \in \Cc^\infty(\Mc)$, $v > 0$, such that $\Rr\left[v^{\frac{4}{n-2}}\gamma\right]$ is constant and negative. $S\left[v^{\frac{4}{n-2}}\gamma\right]$ is constant and negative as well, because of \eqref{SchoutenRicci}. For all $k > 0$, \eqref{Sgammatilde} implies that
\[S\left[k v^{\frac{4}{n-2}}\gamma\right] k = S\left[v^{\frac{4}{n-2}}\gamma\right].\]
We fix
\[k := \dfrac{a_1}{1-a_0}S\left[v^{\frac{4}{n-2}}\gamma\right] > 0\]
by hypothesis. Then
\begin{align*}
a_0 + a_1 S\left[k v^{\frac{4}{n-2}}\gamma\right] &= a_0 + \dfrac{a_1}{k} S\left[v^{\frac{4}{n-2}}\gamma\right] \\
&= 1.
\end{align*}

We set
\begin{align*}
\tilde{\gamma} &:= k v^{\frac{4}{n-2}}\gamma,
&\vec{\alpha} &:= (a_0, a_1, 0, \ldots, 0).
\end{align*}

We use the Theorem \ref{thminv} on $T\left[\tilde{\gamma}\right]$ and deduce that $\derp{T\left[\tilde{\gamma}\right]}{u}(1,\vec{\alpha})$ is invertible. So we can apply the implicit function theorem to $T\left[\tilde{\gamma}\right] : W^{2+\varepsilon,2}(\Mc) \times \Rb^{p+1} \longrightarrow W^{\varepsilon,2}(\Mc)$. We obtain that there exists $\eta > 0$ such that for all $\vec{a} = (a_0, a_1, a_2, \ldots, a_p) \in \Rb^{p+1}$,
\begin{equation} \label{tfiTgammatilde}
\left\|\vec{a} - \vec{\alpha}\right\| < \eta \qquad \Longrightarrow \qquad
\text{there exists $w_{\vec{a}} \in W^{2+\varepsilon,2}(\Mc)$ such that} \quad T\left[\tilde{\gamma}\right]\left(w_{\vec{a}},\vec{a}\right) = T\left[\tilde{\gamma}\right](1,\vec{\alpha}).
\end{equation}
Moreover, $\vec{a} \longmapsto w_{\vec{a}}$ is $\Cc^1$.

\eqref{tfiTgammatilde} means that
\begin{align}
\sum_{k=0}^p a_k \sigma_k\left(\lambda\left(w_{\vec{a}}^{\frac{4}{n-2}}\tilde{\gamma}\right)\right)
	&= T\left[\tilde{\gamma}\right]\left(w_{\vec{a}},\vec{a}\right) \nonumber \\
	&= T\left[\tilde{\gamma}\right](1,\vec{\alpha}) \nonumber \\
	&= a_0 + a_1 S\left[\tilde{\gamma}\right] \nonumber \\
	&= 1. \label{tildefa}
\end{align}

We set
\[u_{\vec{a}} := k^{\frac{n-2}{4}}vw_{\vec{a}}.\]
$u_{\vec{a}} \in W^{2+\varepsilon,2}(\Mc)$, $u_{\vec{a}} > 0$, and
\begin{align*}
\sum_{k=0}^p a_k \sigma_k\left(\lambda\left(u_{\vec{a}}^{\frac{4}{n-2}} \gamma\right)\right)
	&= \sum_{k=0}^p a_k \sigma_k\left(\lambda\left(w_{\vec{a}}^{\frac{4}{n-2}}\tilde{\gamma}\right)\right) \\
	&= 1.
\end{align*}
Moreover, $\vec{a} \longmapsto u_{\vec{a}}$ is $\Cc^1$: this is the statement of the Theorem \ref{tfi}.
\end{proof}

\begin{rem}
If $a_1(a_0-1) < 0$, there is no conformal solution equal to $1$ for a negative scalar curvature.

The hypothesis on the negativity of $\Rr$ unfortunately takes us away from the classical physics case. If $a_1 \approx 1$, the hypothesis $a_1(a_0-1) > 0$ implies that $a_0 > 1$, hence $-2\Lambda > \Ec$: according to the weak energy condition, that is a case of negative cosmological constant. However a negative cosmological constant naturally appears in the AdS/CFT correspondence: see \cite{Camanho13} for an application of Lovelock theories in this context.
\end{rem}

The Theorem \ref{tfi} can immediately be interpreted in terms of the equation \eqref{eqfa}:
\begin{cor}
Let $a_0, a_1 \in \Rb$ such that $a_1(a_0-1)>0$.

Let $\varepsilon > 0$ and $\gamma$ be such that $\Rr[\gamma] < 0$. We set
\[\vec{\alpha} = (a_0, a_1, 0, \ldots, 0) \in \Rb^{p+1}.\]
Then there exists $\eta > 0$ such that for all $\vec{a} = (a_0, a_1, a_2, \ldots, a_p) \in \Rb^{p+1}$,
\[\left\|\vec{a} - \vec{\alpha}\right\| < \eta \qquad \Longrightarrow \qquad
\text{there exists $u_{\vec{a}} \in W^{2+\varepsilon,2}(\Mc)$ such that} \quad
\displaystyle{f_{\vec{a}} \left(\lambda\left(u_{\vec{a}}^{\frac{4}{n-2}}\gamma\right)\right) = 1}.\]
That is to say, $(\ref{eqfa},\vec{a}, 1)$ is solved in $W^{2+\varepsilon,2}(\Mc)$.
\end{cor}

\section{Conclusion}
In the first section we exposed the Lovelock constraint equations and their peculiar form in the case of a locally conformally flat compact manifold in vacuum, and a time-symmetric space-like hypersurface. We showed that the momentum constraint vanishes and the Hamiltonian constraint becomes a generalisation of the $\sigma_k$-Yamabe problem: it amounts to the prescription of an arbitrary linear combination of the $\sigma_k$-curvatures
\[f_{\vec{a}} = a_0 \sigma_0 + a_1 \sigma_1 + \ldots + a_p \sigma_p.\]
We restrict our study to the prescription of a constant.

\vspace{\baselineskip}

In the second section we use the results on this topic that are present in the literature and date from the $2000$'s. They apply to a concave function in the G\aa rding cone: this implies that the scalar curvature of the manifold be positive. Moreover we restrict to non-negative coefficients: $a_k \geq 0$.

We cite the results of the co-paper \cite{Lac17-3} that exhibits some cases in which $f_{\vec{a}}^{1/p}$ is concave. For an arbitrary $p$, it is proved in \cite{Lac17-3} that a sufficient condition for $f_{\vec{a}}^{1/p}$ to be concave is that $a_0, a_1, \ldots, a_p$ are the highest-degree coefficients of some real-rooted polynomial. Then, using Walsh's theorem \ref{thmLac17-3}\ref{thmWalsh}, an other sufficient condition is established: for all $p$, if the restriction of $f_{\vec{a}}$ to the diagonal is real-rooted, then $f_{\vec{a}}^{1/p}$ is concave. The reverse is true for $p=2$.

In this paper we show that the real-rootedness of the restriction of $f_{\vec{a}}$ to the diagonal is equivalent to the real-rootedness of the factorisation of the projected Lovelock sum $f_{\vec{a}}$ in $p$ concircular curvatures. This raises the question of the physical meaning of this concircular factorisation.

Then we show that for $p=2$, for every set of coefficients $a_0, a_1, a_2 \geq 0$, there exists a concavifying function that gives a solution to the Hamiltonian constraint equation. We conjecture that for a map $F$, the concavity of $F \circ f_{\vec{a}}$ in the G\aa rding cone is equivalent to its concavity on the diagonal. We conjecture as well that a concavifying function exists for all $p \geq 2$.

\vspace{\baselineskip}

In the third section we study the other case: when the scalar curvature of the manifold is negative. We use tools from the Yamabe problem and show that the implicit function theorem applies to $f_{\vec{a}}$ at a metric -- of conformal factor $u$ -- of constant scalar curvature. This gives us in a neighbourhood of $u$ a family of conformal solutions for the prescription of $f_{\vec{a}}$.

\vspace{\baselineskip}

There are several ways in which the content of this paper could naturally be extended:
\begin{itemize}
\item to show the conjectures of the second section about the concavifying functions,

\item to understand the physical meaning of the concircular factorisation of the Lovelock sum, and

\item to release one of the many hypothesis on $\Mc$: vacuum, compact, locally conformally flat, no boundary, and time-symmetric.
\end{itemize}

\section*{Acknowledgements}
I would like to deeply thank Pr. Emmanuel Humbert who helped me throughout the development of this paper, especially the third section.

%\nocite{*} %Pour faire apparatre les références non citées explicitement.
%\bibliographystyle{hplain}
\bibliographystyle{plain}
\bibliography{../../Biblio}

\end{document}